\documentclass[10pt]{article}
\usepackage{enumerate}
\usepackage[colorlinks=true, linkcolor = blue, citecolor = cyan]{hyperref}
\usepackage{amsfonts}
\usepackage[dvipsnames]{xcolor}
\usepackage{amsthm,amssymb,amsmath}
\usepackage{graphicx}
\usepackage{tikz}
\usetikzlibrary{arrows}
\usepackage{thmtools}
\usepackage{verbatim}
\usepackage{multirow}
\usepackage[top=3cm,right=3cm,bottom=2.5cm,left=2.5cm]{geometry}         

\author{Doratossadat Dastgheib\footnote{d$_-$dastgheib@sbu.ac.ir} , Hadi Farahani\footnote{h$_-$farahani@sbu.ac.ir}}
\title{Doxastic \L{ukasiewicz} Logic with Public Announcement}

\newtheorem{Proposition}{\textbf{Proposition}}[section]
\newtheorem{Lemma}{\textbf{Lemma}}[section]
\newtheorem{Theorem}{\textbf{Theorem}}[section]

\newtheorem{Remark}{\textbf{Remark}}[section]
\newenvironment{Definition}{\vskip 0.1cm\textbf{Definition.}}{\hfill$\blacktriangleleft$\vskip 0.1cm}

\declaretheoremstyle[headfont=\bfseries]{normalhead}

\numberwithin{equation}{section}
\begin{document}
	\maketitle
	\begin{abstract}
			In this paper we propose a doxastic extension $\textbf{B\L}^+$ of  \L ukasiewicz logic  which is sound and complete relative to  the introduced corresponding  semantics. Also, we equip  our doxastic \L ukasiewicz logic $\textbf{B\L}^+$  with public announcement and propose the logic \textbf{D\L}. As an application, we model a fuzzy version of muddy children puzzle with public announcement using  \textbf{D\L}. Finally, we define  a translation  between  \textbf{D\L} and $\textbf{B\L}^+$, and prove the soundness and completeness theorems for  \textbf{D\L}.
	

	\end{abstract}
\section{Introduction}

Dynamical systems are applied to describe the evolution of a system over time. Dynamical logics are used to verify and specify some properties of dynamical systems.
For example, the propositional dynamic logic (PDL)  describes some properties of programs such as comparing the expressing power of programming constructs  \cite{Fischer1978, Harel1979, Harel2000, Teheux2014}. Even in recent works,  some differential dynamical logics have proposed that study  properties of dynamical systems \cite{Platzer2012, Platzer-essay2012}. 

Since \textit{knowledge} and \textit{belief}  can change over time, the dynamics of these systems  are important. The term \textit{Epistemic logic} was  initially used in \cite{vonWright1951} as a category of modal logics,  named \textit{ modes of knowing}   which was treated by logicians. The first precise definition of the two modalities of knowledge and belief was introduced by Hintikka \cite{Hintikka1962} who used the \textit{Doxastic} term  when we consider belief instead of knowledge operator. Further studies  have been performed with the help of computer scientists and game theorists \cite{Aumann1999,Bonanno1999, Fagin1995-reasoning-about-knowledge,Mayer1995}. Also, in many interdisciplinary areas, this field was developed such as economics \cite{Samuelson2004}, computer security \cite{Ramanujam2005}, multi-agent systems \cite{Halpern1990, van_der_Hoek2003}, and social sciences \cite{Gintis2009, Parikh2002}.

Epistemic dynamic logic not only describes how the information can change over  time, but also  epistemic modalities give the ability to reason about the information itself. \cite{Plaza1989} as one of the pioneers in this area defines a logic for a public announcement that indicates how the knowledge of the agents would change after the public announcement of  a proposition. A Gentzen system for logic of puplic announcement proposed in \cite{Negiri2023} in which even false announcement is possible. In \cite{Lindstrom1999-a, Lindstrom1999-b, Segerberg1999-a, Segerberg1999-b}  dynamic doxastic logic and belief revision are studied.
An extension of Aucher's dynamic belief revision (See \cite{Aucher2006}) to the fuzzy environment introduced in \cite{Jing2014}.
In \cite{Ditmarsch1999, Ditmarsch2001, Ditmarsch2003} some actions  more complex than  public announcements have developed. Also,  \cite{Ditmarsch2016, Kuijer2017} study some dynamic epistemic logics in which accessibility relations instead of the possible states can be  updated.

Since knowledge and belief are  not static over time and they contain a degree of vagueness, it seems that interpreting them via a \textit{fuzzy } perspective is appropriate. 
In \cite{CMRR13, Caicedo2004,CR10,CR15}, some modal extensions of G\"odel fuzzy logics are proposed, also some modal extensions of \L ukasiewicz logic  and product fuzzy logic are introduced in \cite{Hansoul2006, Hansoul2013, product2015}.

\cite{Osterman1988} introduces a many-valued modal propositional calculi and gives a decision procedure for proving truth formulae.
In \cite{Santos2020} a four-valued dynamic epistemic logic is proposed, then using a tableau system its soundness and completeness are shown. In \cite{DiNola2021} a dynamic $n$-valued \L ukasiewicz logic ID\L$_n$ was introduced  and by proposing a Kripke-based semantics, some applications of ID\L$_n$ in immune systems are studied. A forensic dynamic n-valued  \L ukasiewicz logic and its corresponding forensic dynamic MV$_n$\_algebra have proposed in  \cite{DiNola2022}. In \cite{DiNola2020} a \L ukasiewicz extension of dynamic propositional logic; introduced in \cite{Kozen1980, Segerberg1977}; is proposed, also a corresponding dynamic MV-algebra is defined. Some possible definitions of public announcement  for G\"odel modal logic are investigated in \cite{Pischke2021}. A fuzzy epistemic logic with public announcement in which both fuzzy transitions and fuzzy propositions are  introduced in \cite{Benevides2022} such that the corresponding semantics has defined using G\"odel algebra. 
Also, a doxastic extension of fuzzy \L ukasiewicz logic \textbf{B\L} based on a pseudo-classical belief is introduced in \cite{Dll2022}.


In this paper, at first we propose an axiomatic system $\textbf{B\L}^+$ that is an extension of  \textbf{B\L} introduced in \cite{Dll2022} which its language is an expansion of the language of \textbf{B\L} with a new operator
	$\succeq$ to  compare the fuzzy  validity of a formula with a fixed value.  A formula $\varphi\succeq g$ is valid when the  validity of $\varphi$ is at least  g, where  g$\in [0,1]$.  By adding this operator to the language, we assign a crisp value to the amount of truth  of a given formula, and intuitively we decide whether a formula has the desired fuzzy value or not. So a formula containing $\succeq$ can be viewed as a classical formula. We also prove the soundness and completeness theorems for $\textbf{B\L}^+$.
	
	 	Afterward we introduce a  dynamic doxastic logic \textbf{D\L} over $\textbf{B\L}^+$. This logic is an extension of doxastic \L ukasiewicz logic which is equipped with a public announcement operator. We suppose that every public announcement is true and explicitly gives us some information about a formula. Then we model a fuzzy version of the muddy children puzzle using this logic. Furthermore, we define  a translation from \textbf{D\L} to $\textbf{B\L}^+$ and prove the soundness and completeness of \textbf{D\L}.
	


\section{Modified  doxastic \L ukasiewicz logic}
In this section we first review  some axioms and properties of propositional fuzzy \L ukasiewicz logic that are used throughout  this paper (see \cite{Hajek1998} for more details).  Then, we review some definitions  from  \cite{Dll2022} to propose modified doxastic \L ukasiewicz logic \textbf{B\L}$^+$ and show that it is sound and complete corresponding to the desired semantics. 

The propositional \L ukasiewicz logic is an extension of fuzzy Basic Logic \textbf{BL} with double negation axiom $\neg \neg \varphi \rightarrow \varphi$. The following propositions are valid in \L ukasiewicz logic: 


$$\begin{array}{llcll}
	(A1) &(\varphi \, \& \, \psi)\rightarrow \varphi &\quad &(A2)& (\varphi \, \& \, \psi) \rightarrow (\psi \, \& \, \varphi) \\
	(\L 0)& \neg \neg \varphi \leftrightarrow \varphi &\quad& (\L 1)& (\neg \varphi \rightarrow \neg \psi) \rightarrow (\psi \rightarrow \varphi) \\
	
	(\L 2) & 
	\multicolumn{4}{l}{((\varphi_1\rightarrow \psi_1) \,\&\, (\varphi_2\rightarrow \psi_2)) \rightarrow ((\varphi_1\,\&\,\varphi_2)\rightarrow (\psi_1\,\&\,\psi_2))}\\
	
%
\end{array}$$

Throughout this paper we denote the set of atomic propositions and the set of agents by $\mathcal{P}$ and $\mathcal{A}$ respectively, and use notation $\perp$  for the atomic proposition that always takes value $0$.

In \cite{Dll2022}, two classes of doxastic extensions of fuzzy \L ukasiewicz logic have proposed. One class is equipped with pseudo-classical that has properties similar to the classical belief and the other class is based on a new notion of belief that is called skeptical belief.
 In the following
   we expand the language of pseudo-classical belief using a new operator $\succeq$ which compares the fuzzy  validity of a given formula with some fixed value and helps us to decide whether a formula has a desired fuzzy value or not. For a  number g$\in [0,1]$, the formula $\varphi\succeq g$ is valid when the validity of $\varphi$ is at least g.  

%
%

\begin{Definition}
	The language of \textit{Modified  doxastic \L ukasiewicz logic}  denoted by \textbf{D\L L}$^+$ is defined with the following BNF:
	$$\varphi::= \perp \;|\; p \;|\; \neg \varphi\;|\; \varphi \succeq g \;|\; \varphi \, \& \, \varphi \;|\; \varphi \rightarrow \varphi \;|\; B_a \varphi \\$$
	where $p\in\mathcal{P}$, $a\in \mathcal{A}$ and $g\in [0,1]$ is a rational number. Throughout this paper whenever we use $g\in [0,1]$ we mean $g$ is a rational number from the interval $[0,1]$. The other usual connectives $\vee$, $\wedge$ and $\veebar$ are defined similar as \L ukasiewicz logic.
\end{Definition}
	
	In the following we give the definition of a  Doxastic \L ukasiewicz logic model proposed in \cite{Dll2022}.
\begin{Definition} 
	A \textit{Doxastic \L ukasiewicz Logic model} (or in short \textit{D\L L-model}) is a tuple $\mathfrak{M}= (S, r_{a_{|a\in \mathcal{A}}}, \pi)$ in which $S$ is a set includes the states of  the model, $r_{a_{|a\in A}}: S\times S \rightarrow [0,1]$ is indistinguishability function and $\pi: S\times \mathcal{P} \rightarrow [0,1]$ is the valuation function that assigns values to each proposition in each state.

	Suppose that $\mathfrak{M} = (S, r_{a_{|a\in \mathcal{A}}}, \pi)$ is a  D\L L-model.
	%
	For each formula $\varphi$ in  \textbf{D\L L}$^+$  and each state $s\in S$ we denote $V(s,\varphi)$ as an extended valuation function which is defined recursively as follows. For simplicity we use $V_s(\varphi)$ instead of $V(s,\varphi)$ and use superscript $\mathfrak{M}$ like $V^{\mathfrak{M}}$ to emphasis on the model if it is needed.
	\begin{align*}
		& V_s^{\mathfrak{M}}(p) = \pi(s,p) \;\; \forall p\in \mathcal{P}, & \\
		& V_s^{\mathfrak{M}}(\neg \varphi) = 1- V_s^{\mathfrak{M}}(\varphi), & \\
		& V_s^\mathfrak{M}(\varphi \succeq g)= \left \lbrace \begin{array}{ll}
			1 & V_s^\mathfrak{M}(\varphi)\geq g \\
			0 & V_s^\mathfrak{M}(\varphi) < g,
		\end{array} \right. & \\
		& V_s^{\mathfrak{M}}(\varphi \,\&\, \psi) = \max\{0, V_s^{\mathfrak{M}}(\varphi)+ V_s^{\mathfrak{M}}(\psi) - 1\}, & \\
		& V_s^{\mathfrak{M}}(\varphi \rightarrow \psi) = \min \{1, 1- V_s^{\mathfrak{M}}(\varphi)+V_s^{\mathfrak{M}}(\psi)\},& \\
		& V_s^{\mathfrak{M}}(B_a \varphi) = \inf_{s'\in S} \max\{1-r_a(s,s'),V_{s'}^{\mathfrak{M}}(\varphi)\}, & 
	\end{align*}
\end{Definition}
\begin{Remark}
	Note that the formula $\varphi \succeq g$ takes crisp values to the amount of truth value of a given formula
\end{Remark}

\begin{Definition}
	Let $\varphi$ be a formula and $\mathfrak{M} = (S, r_{a_{|a\in \mathcal{A}}}, \pi)$ be a D\L L-model and $s\in S$. We say that $\varphi$ is \textit{valid in a pointed model} $(\mathfrak{M},s)$ if $V_s(\varphi)=1$ and denote it by $(\mathfrak{M},s)\vDash \varphi$. If for all $s\in S$ we have $(\mathfrak{M},s)\vDash \varphi$, then we call it $\mathfrak{M}$\textit{-valid} and use $\mathfrak{M} \vDash \varphi$ to denote it. If for all models $\mathfrak{M}$ in a class of models $\mathcal{M}$, the formula is $\mathfrak{M}$\textit{-valid}, we show it by $\mathcal{M}\vDash \varphi$ and call it $\mathcal{M}$-valid. We use $\vDash \varphi$ notation if for all models $\mathfrak{M}$ we have $\mathfrak{M} \vDash \varphi$ and call $\varphi$ as a \textit{valid} formula.
\end{Definition}

\begin{Proposition} \label{prop_Lg_validity}
	The following statements are valid.
	\begin{itemize}
		\item $(\varphi \,\&\, \psi)\succeq g \rightarrow (\varphi\succeq g \,\&\, \psi \succeq g)$
		\item $(\varphi \succeq g \,\&\, \psi \succeq g') \rightarrow (\varphi \succeq g'' \,\&\, \psi \succeq g')\qquad \text{s. t.}\quad g\geq g'' $
	\end{itemize}
\end{Proposition}

Similar to the axiomatic system \textbf{B\L} described in \cite{Dll2022}, we introduce the following axiomatic system \textbf{B\L}$^+$, where $\varphi$ and $\psi$ are \textbf{D\L L}$^+$-formulas:
\begin{itemize}
	\item[(\L$_B$0)] All instances of tautologies in propositional \L ukasiewicz logic
	\item[(\L$_B$1)] $(B  \varphi \, \& \, B (\varphi \rightarrow \psi)) \rightarrow (B  \psi)$
	\item[(\L$_B$2)] $\neg B  \perp$
	\item[(\L$_g$0)] $(\varphi \,\&\, \psi)\succeq g \rightarrow (\varphi\succeq g \,\&\, \psi \succeq g)$
	\item[(\L$_g$1)] $(\varphi \succeq g \,\&\, \psi \succeq g') \rightarrow (\varphi \succeq g'' \,\&\, \psi \succeq g')\qquad \text{s. t.}\quad g\geq g'' $
	\item[(R$_{\text{MP}}$)] $\cfrac{\varphi \qquad \varphi \rightarrow \psi}{\psi}$
	\item[(R$_\text{B}$)] $\cfrac{\varphi}{B  \varphi}$ 
	\item[(R$_\text{G}$)] $\cfrac{\varphi}{  \varphi\succeq g}$
\end{itemize}

As we see, the differences between \textbf{B\L} and \textbf{B\L}$^+$ are axioms (\L$_g$0) and (\L$_g$1) and  the rule (R$_\text{G}$).
\begin{Lemma} \label{gb_admissible_rules}
	The inference rules (R$_{\text{MP}}$), (R$_\text{B}$) and (R$_\text{G}$)  are semantically admissible, that is if the premises of (R$_{\text{MP}}$), (R$_\text{B}$) or (R$_\text{G}$) are valid, then their conclusions are valid.
\end{Lemma}
\begin{proof}
	It's obvious by definition.
\end{proof}


\begin{Theorem}\label{soundness_pseudo-classical} \textbf{(Soundness)}
	Let $\mathcal{M}$ be a class of D\L L$^{}$-models.
	If $\; \vdash_{\textbf{B\L}^+ }\varphi$, then $\mathcal{M}\vDash \varphi$.
\end{Theorem} 
\begin{proof}
	The proof is obtained straightforwardly by using Proposition 3.1 in \cite{Dll2022}, Proposition \ref{prop_Lg_validity} and Lemma \ref{gb_admissible_rules}.
\end{proof}
\begin{Definition}
	A \textbf{D\L L}$^+$-formula $\varphi$ is called \textit{\textbf{B\L}$^{+}$-consistent}, if $\nvdash_{\textbf{B\L}^+}\neg \varphi$. A finite set $\{\varphi_1, \cdots, \varphi_n\}$ is \textbf{B\L}$^{+}$-consistent if $\varphi_1\,\&\,\cdots\,\&\,\varphi_n$ is \textbf{B\L}$^{+}$-consistent. An infinite set $\Gamma$ of \textbf{D\L L}$^{+}$ is \textbf{B\L}$^{+}$-consistent, if all of its finite subsets are \textbf{B\L}$^{+}$-consistent. If the following conditions hold for $\Gamma$, we call $\Gamma$ a \textit{maximal and \textbf{B\L}$^{+}$-consistent} set:
	\begin{enumerate}
		\item  $\Gamma$ be a $\textbf{B\L}^+$-consistent set.
		\item For all D\L L$^+$-formula $\psi\notin\Gamma$, the set $\Gamma \cup \{\psi\}$   not to be \textbf{B\L}$^+$-consistent.
	\end{enumerate}
\end{Definition}

In the following, Lemma \ref{maximal}, Lemma \ref{and_Luka} and Theorem \ref{consistent_e} have similar proofs shown in section 3 of \cite{Dll2022}.

\begin{Lemma}\label{maximal}
	Let $\textbf{B\L}^{+}$ be an axiomatic system. We have
	\begin{itemize}
		\item[(i)] Each $\textbf{B\L}^{+}$-consistent set $\Phi$ of \textbf{D\L L}$^+$-formulae can be extended to a maximal $\textbf{B\L}^{+}$-consistent set.
		\item[(ii)] If $\Phi$ is a maximal  $\textbf{B\L}^{+}$-consistent set, then for all \textbf{D\L L}$^+$-formulae $\varphi$ and $\psi$:
		\begin{enumerate}
			\item $\varphi\,\&\,\psi \in \Phi$ if and only if $\varphi \in \Phi$ and $\psi\in \Phi$,
			\item If $\varphi \in \Phi$ and $\varphi \rightarrow \psi \in \Phi$, then $\psi \in \Phi$,
			\item If $\vdash_{\textbf{B\L}^+} \varphi$, then $\varphi \in \Phi$,
			\item $\varphi \in \Phi$ or $\neg \varphi \in \Phi$. 
		\end{enumerate}
	\end{itemize}
\end{Lemma}
\begin{proof}
	The proof is similar to the proof of Lemma 3.4 in \cite{Dll2022}.
\end{proof}

\begin{Lemma}\label{and_Luka}
	Let $\Gamma = \{\varphi_1, \cdots, \varphi_n\}$ be a set of formulae, then $\Gamma \vdash_{\textbf{B\L}^+} \varphi_1\,\&\,\cdots\,\&\,\varphi_n$.
\end{Lemma}
\begin{proof}
	The proof is similar to the proof of Lemma 3.5 in \cite{Dll2022}.
\end{proof}

\begin{Theorem}\label{consistent_e}
	Let $\Phi$ be a  \textbf{B\L}$^{+}$-consistent set of D\L L$^{+}$-formulae and $\varphi$ be a D\L L$^{+}$-formula such that 
	$\Phi \nvdash_{\textbf{B\L}^{+}} \varphi$. If $\Phi^{*} = \Phi\cup \{\neg \varphi\}$, then $\Phi^{*}$ is \textbf{B\L}$^{+}$-consistent.
\end{Theorem}
\begin{proof}
	The proof is similar to the proof of Theorem 3.2 in \cite{Dll2022}.
\end{proof}

\begin{Theorem} \label{thm_model_existence}
	Let $\Phi$ be a \textbf{B\L}$^{+}$-consistent set and $\Phi \vdash_{\textbf{B\L}^+} \varphi$, where $\varphi$ is a D\L L$^+$-formula. Then there is a D\L L-model $\mathfrak{M}$ and a state $s$ such that $V_{s}^{\mathfrak{M}}(\varphi) = 1$. 
\end{Theorem}

\begin{proof}
	We just need to modify the proof of Theorem 3.3 in \cite{Dll2022} slightly.
%
	%
	 It is needed to  show that for the following canonical model $\mathfrak{M}^c = (S^c, r^c, \pi^c)$, and for each maximal and $\textbf{B\L}^{+}$-consistent set $\Phi^*$ we have:
	\begin{equation}\label{eq_00_canonical}
		\varphi\in \Phi^* \iff V_{s_{\Phi^*}}^{\mathfrak{M}^c}(\varphi)=1.
	\end{equation}
where, the canonical model $\mathfrak{M}^c$ is defined as follows: 
	$$S^c = \left\lbrace \begin{array}{l|cl}
		\multirow{3}{*}[0.25cm]{$s_{\Phi^{*}}$}& \multirow{3}{*}{} & \Phi^* \text{ is a maximal and }\textbf{B\L}^+\text{-consistent set containing a maximal }   \textbf{B\L}\text{-consistent }\\
			& &  \text{set }\Phi  \text{ and all formulae of the form } \varphi\succeq g,  \text{where }  \varphi \in \Phi,\; g\in(0,1]
		\end{array}\right\rbrace,
		$$

	\begin{align*}
		& r ^c (s_\Phi , s_\Psi) = \left\lbrace \begin{array}{lc}
			1 & \Phi\backslash B  \subseteq \Psi \\
			0 & \text{otherwise}
		\end{array}\right.; \qquad \Phi \backslash B  \stackrel{def}{=} \{\varphi \,|\, B  \varphi \in \Phi\},\\
		& \pi^c (s_{\Phi}, p) = \left\lbrace \begin{array}{lc}
			1 & p\in \Phi \\
			0 & p \notin \Phi 
		\end{array} \right.; \qquad p\in \mathcal{P}.
	\end{align*}

	By induction on the complexity of $\varphi$, we show that statement (\ref{eq_00_canonical}) holds. It is  enough to check only the case $\varphi  = \psi \succeq g$, since the other cases have similar proof as the proof of Theorem 3.3 in \cite{Dll2022}.
	
	For ``$\Rightarrow$" direction, if $\psi \succeq g \in \Phi^*$, then by the definition of $S^c$, there is a maximal and \textbf{B\L}-consistent set $\Phi$, such that
	$\psi \in \Phi\subseteq \Phi^*$, thus  $\psi \in \Phi^*$, therefore by induction hypothesis we have $V_{s_{\Phi^*}}^{\mathfrak{M}^c}(\psi)=1$. Hence $V_{s_{\Phi^*}}^{\mathfrak{M}^c}(\psi)\geq g$, which means that $V_{s_{\Phi^*}}^{\mathfrak{M}^c}(\psi\succeq g)=1$.
	For the other direction, assume $V_{s_{\Phi^*}}^{\mathfrak{M}^c}(\psi\succeq g)=1$. Thus $V_{s_{\Phi^*}}^{\mathfrak{M}^c}(\psi)\geq g$. Since all formulas in the canonical model take crisp values and by definition we have $g\in(0,1]$, we obtain $V_{s_{\Phi^*}}^{\mathfrak{M}^c}(\psi)=1$. Therefore, by induction hypothesis we have $\psi \in \Phi^*$, and so by definition and maximality of $\Phi^*$ and (R$_\text{G}$) we have $\psi\succeq g \in \Phi^*$.
\end{proof}

Note that the set of states $S^c$ of canonical model $\mathfrak{M}^c$ defined in the proof of Theorem \ref{thm_model_existence} is non-empty by following lemma.

\begin{Lemma}\label{lem_g_set_consistency}
	If $\Phi$ is a \textbf{B\L}$^+$-consistent, then $\{\varphi\succeq g\mid \varphi\in \Phi,\; g\in(0,1])\}$ is \textbf{B\L}$^+$-consistent too.
\end{Lemma}
\begin{proof}
	For the sake of contradiction assume $\Phi$ is \textbf{B\L}$^+$-consistent but $\Gamma = \{\varphi\succeq g\mid \varphi\in \Phi,\; g\in(0,1]\}$ is not \textbf{B\L}$^+$-consistent. So by the definition, there is a finite subset $\{\varphi_1\succeq g_1, \cdots, \varphi_n\succeq g_n\}\subseteq \Gamma$ such that
	\begin{equation}\label{eq_inconsistency}
		\vdash_{\textbf{B\L}^+} \neg (\varphi_1\succeq g_1\,\&\, \cdots \,\&\,\varphi_n\succeq g_n).
	\end{equation}
	Let $g_0 = \max\{g_1,\cdots, g_n\}$ and $\Gamma' = \{\varphi_1, \cdots, \varphi_n\}$. By  \textbf{B\L}$^+$-consistency of $\Gamma'$ and Lemma \ref{and_Luka} we have:
		\begin{align*}
		&(1)&& \Gamma' \vdash_{\textbf{B\L}^+} \varphi_1 \,\&\,\cdots \,\&\, \varphi_n & \\
		&(2)&& \Gamma' \vdash_{\textbf{B\L}^+} (\varphi_1\,\&\,\cdots \,\&\,\varphi_n) \succeq g_0 & (R_G)\\
		&(3)&& \Gamma' \vdash_{\textbf{B\L}^+} ((\varphi_1\,\&\,\cdots \,\&\,\varphi_n) \succeq g_0) \rightarrow (\varphi_1 \succeq g_0 \,\&\,\cdots \,\&\, \varphi_n\succeq g_0)& (\L_g0) \\ 
		&(4)&&\Gamma' \vdash_{\textbf{B\L}^+} \varphi_1 \succeq g_0 \,\&\,\cdots \,\&\, \varphi_n\succeq g_0 & (2),(3), (R_{MP})\\
		&(5)&&\Gamma' \vdash_{\textbf{B\L}^+} (\varphi_1 \succeq g_0 \,\&\,\cdots \,\&\, \varphi_n\succeq g_0) \rightarrow (\varphi_1 \succeq g_1 \,\&\,\cdots \,\&\, \varphi_n\succeq g_0)& (\L_g1) \\ 
		&(6)&&\Gamma'\vdash_{\textbf{B\L}^+}(\varphi_1 \succeq g_1 \,\&\,\cdots \,\&\, \varphi_n\succeq g_0)& (4),(5),(R_{MP})\\
		&\vdots&&& \\
		&(i)&&\Gamma' \vdash_{\textbf{B\L}^+} (\varphi_1 \succeq g_1 \,\&\,\cdots \,\&\, \varphi_n\succeq g_n)& \text{similar to (5),(6) steps}\\
		\end{align*}
	which has a contradiction to $\Gamma'\vdash_{\textbf{B\L}^+}\neg (\varphi_1\succeq g_1\,\&\, \cdots \,\&\,\varphi_n\succeq g_n)$. Thus the statement holds.
\end{proof}
\begin{Theorem}
	If $\vDash \varphi$, then $\vdash_{\textbf{B\L}^+}\varphi$.
\end{Theorem}
\begin{proof}
	For the sake of contradiction assume we have $\nvdash_{\textbf{B\L}^{+}}\varphi$. By Theorem \ref{consistent_e}, the set $\{\neg \varphi\}$ is \textbf{B\L}$^+$-consistent, and by Lemma \ref{maximal} there is a maximal and \textbf{B\L}$^+$-consistent set $\Phi^*$ that contains $\{\neg \varphi\}$, and thus by Theorem \ref{thm_model_existence} there is a model $\mathfrak{M}$ and state $s$ such that $V_s^\mathfrak{M}(\neg \varphi)=1$ which is a contradiction with $\vDash\varphi.$
\end{proof}

\section{Dynamic Doxastic \L{ukasiewicz Logic} with Public Announcement}

In this section we propose a dynamic version of  doxastic \L ukasiewicz logic. We expand the language \textbf{D\L L}$^+$ with a new operator $[\varphi\succeq g]\psi$ for public announcement. We also propose an axiomatic system  \textbf{D\L} and prove its  soundness and completeness.

\begin{Definition}

		The \textit{language of dynamic doxastic \L ukasiewicz logic} denoted by  \textbf{DD\L L}  is defined by the following BNF.  \textbf{DD\L L} has a new formulae $[\varphi\succeq g] \varphi$, 
		where $\varphi$ is \textbf{D\L L}$^+$-formula.

		$$\varphi::= \perp \;|\; p \;|\; \neg \varphi\;|\; \varphi \succeq g \;|\; \varphi \, \& \, \varphi \;|\; \varphi \rightarrow \varphi \;|\; B_a \varphi \;|\; [\varphi\succeq g] \varphi  \\$$

%
	
	Furthermore, we also use notation $\varphi \leftrightarrow \psi$ for the formula $\varphi\rightarrow \psi \,\&\, \psi\rightarrow \varphi$. We read the new formula $[\varphi\succeq g]\psi$ as \textit{after public announcement of $\varphi\succeq g$, the $\psi$ holds (See example \ref{example_muddy_children}).}
\end{Definition}

	We first give the definition of an \textit{update model}.
	Update model explains how a public announcement $\varphi\succeq g$ effects on the states of the model $\mathfrak{M}$. Note that we assume the public announcements are true, i.e, just the true propositions would be announced and no one announces a lie.

%

\begin{Definition}(\textbf{Update Model})
	Let $\mathfrak{M} = (S, r_{a_{|a\in\mathcal{A}}}, \pi)$  be a D\L L-model. The update model of $\mathfrak{M}$ after announcement $\varphi\succeq g$ is denoted by $\mathfrak{M}^{\varphi\succeq g} = (S^{{\varphi\succeq g}}, r_{a_{|_{a\in \mathcal{A}}}}^{\varphi\succeq g}, \pi^{\varphi\succeq g})$, and is defined as follows:
	\begin{align*}
		& S^{\varphi \succeq g}=\{s\,|\, s\in S^\mathfrak{M}, V_s^{\mathfrak{M}}(\varphi\succeq g)=1\},\\
		& r_a^{\varphi \geq g}(s,s') = r_a^{\mathfrak{M}}(s,s'), \qquad \forall s,s'\in S^{\varphi \succeq g},\\
		& \pi^{\varphi \succeq g} (s,p) = \pi^\mathfrak{M}(s,p),\qquad\forall s\in S^{\varphi\succeq g}, \; p \in \mathcal{P}.
	\end{align*}
Here we define the valuation function for  $[\varphi\succeq g]\psi$, the other cases are defined as in previous section.

	\begin{align*}
		V_s^{\mathfrak{M}}([\varphi \succeq g]\psi) = \left\lbrace \begin{array}{ll}
			1 & V_s(\varphi \succeq g) = 0\\
			V^{\mathfrak{M}^{\varphi \succeq g}} (\psi) & otherwise.
		\end{array} \right. & 
	\end{align*}
\end{Definition}


In the following we model a fuzzy version of muddy children puzzle with public announcement.

\subsection{Example (Fuzzy Muddy Children)}\label{example_muddy_children}

In traditional muddy children puzzle, a group of  children has been playing outdoors and some of them  have become dirty and may have  mud on their foreheads. Children can
just see whether other children are muddy, and not if there is any mud on their
own foreheads.
In \cite{EGL2020} a fuzzy version of muddy children puzzle has proposed, and the authors modeled that using Epistemic G\"odel Logic.
In order to consider fuzzy relations between different states, they suppose that the agents have visual impairment and propose a distinguishing criteria corresponding to the amounts of mud on the agents’ foreheads. 

Here we give a dynamic simplified version of this fuzzy muddy children puzzle. Assume there are three children Alice, Bob and Cath. Each state of the model is a tuple $(m_a, m_b, m_c)$, where $m_a, m_b$ and $m_c$ take fuzzy values in $[0,1]$ corresponding to the amount of mud on the Alice, Bob and Cath's foreheads, respectively. So, the number of all possible states is infinite. In order to see  the effect of some announcements on the model, we restrict the general model to some states shown in   figure \ref{fig_fuzzy_muddy_children}-\textbf{I}.
Accessibility relations for Alice, Bob and Cath are shown by red, green and blue numbers, respectively and it is assumed that the model is reflexive and symmetric. For example, $r_c(s_1,s_2) = 0.95$ states that Cath can  distinguish a little that Bob is a bit muddy (e.g. $V_{s_1}(m_b)= 0.1$) or is slightly more muddy (e.g. $V_{s_2}(m_b) = 0.2$). 

Suppose that the father publicly announces that at least one of the children's forehead is very muddy. For example he publicly announces that $[\neg (\neg m_a \,\&\, \neg m_b \,\&\, \neg m_c)\succeq 0.8]$.
State $s_1$ in which $m_a, m_b$ and $m_c$ have values less than $0.8$ would be removed in the updated model (Figure \ref{fig_fuzzy_muddy_children}-\textbf{II}). 
%
 Also assume after  the first announcement, children have no reaction and so father  publicly  announces the previous  statement for the second time. Then the possible worlds reduce to four states (Figure \ref{fig_fuzzy_muddy_children}-\textbf{III}). In the following we can see the Cath's belief about her muddiness at the beginning, after the first announcement and after the second announcement:

$$\begin{array}{l}
	V_{s_3}^{\mathfrak{M}^{\textbf{I}}}(B_c m_c) = \\
	\qquad \inf\{\max\{0, 0.9\}, \max\{0.05,0.2\}, \max\{1, 0.9\}, \max\{1, 0.9\}\}= 0.2, \\
	V_{s_3}^{\mathfrak{M}^{\textbf{II}}}([\neg (\neg m_a \,\&\, \neg m_b \,\&\, \neg m_c)\succeq 0.8]B_c m_c) = 0.2,\\
	V_{s_3}^{\mathfrak{M}^{\textbf{III}}} ([\neg (\neg m_a \,\&\, \neg m_b \,\&\, \neg m_c)\succeq 0.8][\neg (\neg m_a \,\&\, \neg m_b \,\&\, \neg m_c)\succeq 0.8]B_c m_c)= \\
	\qquad \inf \{\max\{0, 0.9\}, \max\{1, 0.9\}\} = 0.9.
\end{array}$$
The above computation shows that Cath is almost certain of her muddiness after the second announcement.

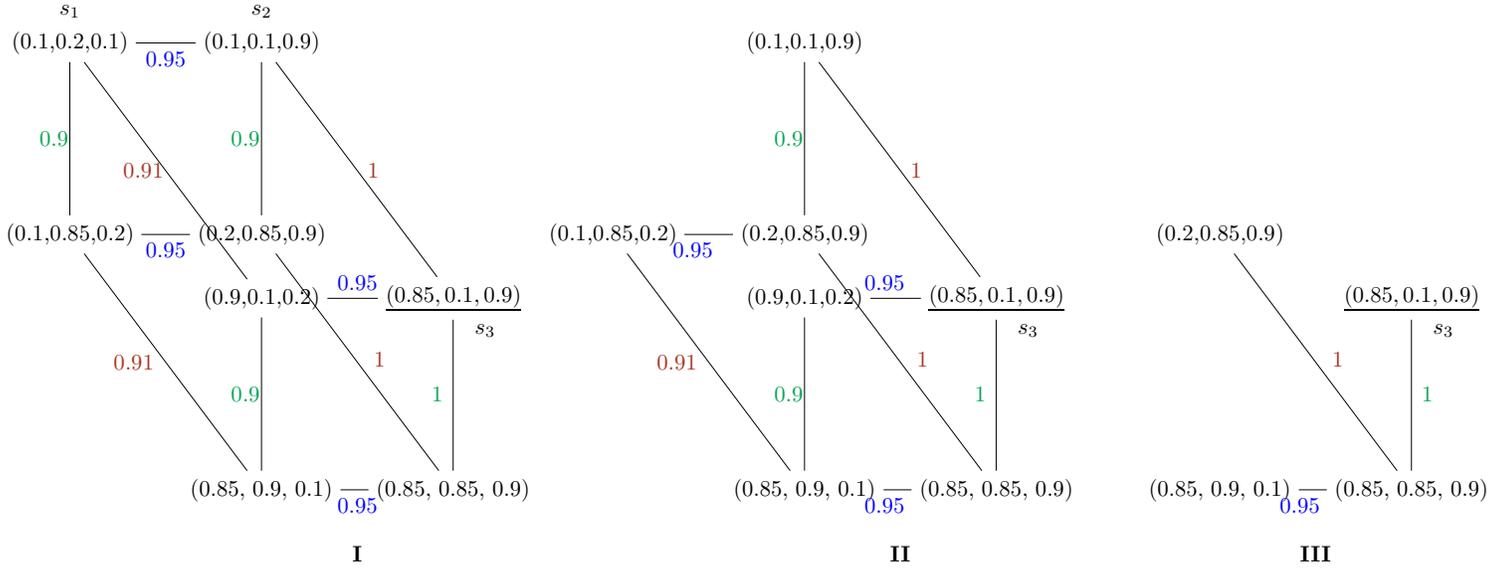
\begin{figure} \label{fig_fuzzy_muddy_children}
	\begin{center}
		\advance\leftskip-1.7cm
		\resizebox*{20cm}{!}{
			\begin{tikzpicture}
				
				\node (v1) at (-1,1) {(0.1,0.2,0.1)};
				\node (v3) at (-1,-2) {(0.1,0.85,0.2)};
				\node (v4) at (2,-2) {(0.2,0.85,0.9)};
				\node (v2) at (2,1) {(0.1,0.1,0.9)};
				
				\node (v5) at (2,-3) {(0.9,0.1,0.2)};
				\node (v8) at (5,-3) {$\underline{(0.85, 0.1, 0.9)}$};
				\node (v6) at (2,-6) {(0.85, 0.9, 0.1)};
				\node (v7) at (5,-6) {(0.85, 0.85, 0.9)};
				
				\draw  (v1) edge (v2);
				\draw  (v1) edge (v3);
				\draw  (v2) edge (v4);
				\draw  (v3) edge (v4);
				\draw  (v5) edge (v6);
				\draw  (v6) edge (v7);
				\draw  (v7) edge (v8);
				\draw  (v8) edge (v5);
				\draw  (v1) edge (v5);
				\draw  (v2) edge (v8);
				\draw  (v3) edge (v6);
				\draw  (v4) edge (v7);
				
				\node at (-1.25,-0.5) {\textcolor{Green}{0.9}};
				\node at (1.75,-0.5) {\textcolor{Green}{0.9}};
				\node at (1.75,-4.5) {\textcolor{Green}{0.9}};
				\node at (4.75,-4.5) {\textcolor{Green}{1}};
				
				\node at (3.5,-6.25) {\textcolor{blue}{0.95}};
				\node at (3.5,-2.75) {\textcolor{blue}{0.95}};
				\node at (0.5,-2.25) {\textcolor{blue}{0.95}};
				\node at (0.5,0.75) {\textcolor{blue}{0.95}};
				
				\node at (3.75,-1) {\textcolor{Mahogany}{1}};
				\node at (0.15,-1) {\textcolor{Mahogany}{0.91}};
				\node at (0,-4) {\textcolor{Mahogany}{0.91}};
				\node at (3.85,-3.95) {\textcolor{Mahogany}{1}};
				
				\node (v10) at (7.5,-2) {(0.1,0.85,0.2)};
				\node (v11) at (10.5,-2) {(0.2,0.85,0.9)};
				\node (v12) at (10.5,1) {(0.1,0.1,0.9)};
				
				\node (v13) at (10.5,-3) {(0.9,0.1,0.2)};
				\node (v14) at (13.5,-3) {$\underline{(0.85, 0.1, 0.9)}$};
				\node (v15) at (10.5,-6) {(0.85, 0.9, 0.1)};
				\node (v16) at (13.5,-6) {(0.85, 0.85, 0.9)};
				
				\draw  (v12) edge (v11);
				\draw  (v10) edge (v11);
				\draw  (v13) edge (v15);
				\draw  (v15) edge (v16);
				\draw  (v16) edge (v14);
				\draw  (v14) edge (v13);
				\draw  (v12) edge (v14);
				\draw  (v10) edge (v15);
				\draw  (v11) edge (v16);
				
				\node at (10.25,-0.5) {\textcolor{Green}{0.9}};
				\node at (10.25,-4.5) {\textcolor{Green}{0.9}};
				\node at (13.25,-4.5) {\textcolor{Green}{1}};
				
				\node at (11.75,-6.25) {\textcolor{blue}{0.95}};
				\node at (11.75,-2.75) {\textcolor{blue}{0.95}};
				\node at (8.75,-2.25) {\textcolor{blue}{0.95}};
				
				\node at (12.25,-1) {\textcolor{Mahogany}{1}};
				
				\node at (8.5,-4) {\textcolor{Mahogany}{0.91}};
				\node at (12.35,-3.95) {\textcolor{Mahogany}{1}};
				
				\node (v17) at (17,-2) {(0.2,0.85,0.9)};
				\node (v18) at (20,-3) {$\underline{(0.85, 0.1, 0.9)}$};
				\node (v19) at (17,-6) {(0.85, 0.9, 0.1)};
				\node (v20) at (20,-6) {(0.85, 0.85, 0.9)};
				
				\draw  (v19) edge (v20);
				\draw  (v20) edge (v18);
				\draw  (v17) edge (v20);
				
				\node at (20.25,-4.5) {\textcolor{Green}{1}};
				\node at (18.25,-6.25) {\textcolor{blue}{0.95}};
				\node at (18.85,-3.95) {\textcolor{Mahogany}{1}};
				
				\node at (5.5,-3.5) {$s_3$};
				\node at (14,-3.5) {$s_3$};
				\node at (20.5,-3.5) {$s_3$};
				\node at (-1,1.5) {$s_1$};
				\node at (2,1.5) {$s_2$};
				
				\node at (3.5,-7) {\textbf{I}};
				\node at (12,-7) {\textbf{II}};
				\node at (18.5,-7) {\textbf{III}};
				
		\end{tikzpicture}}
	\end{center}
	\caption{Fuzzy muddy children example (The state in the real world is shown with an underline)}
\end{figure}

\subsection{Soundess and Completeness} \label{subsection-pa-soundness_and_completeness}
In this section we propose an axiomatic system \textbf{D\L}, and show that it is sound. Then we define a translation between $\textbf{B\L}^+$ and $\textbf{D\L}$ that is used to prove the completeness theorem. 

\begin{Proposition} \label{valid_schema}
	The following schemata are valid.
	\begin{enumerate}
		\item $[\varphi \succeq g]p \leftrightarrow ((\varphi \succeq g) \rightarrow p)$
		\item $[\varphi \succeq g]\neg \psi \leftrightarrow ((\varphi \succeq g) \rightarrow \neg [\varphi \succeq g]\psi)$
		\item $[\varphi \succeq g](\psi \,\&\, \chi) \leftrightarrow ([\varphi \succeq g]\psi \,\&\, [\varphi \succeq g]\chi)$
		\item $[\varphi \succeq g] (\psi \rightarrow \chi) \leftrightarrow [\varphi \succeq g]\neg (\psi \,\&\, \neg \chi)$
		\item $[\varphi \succeq g]B_a \psi \leftrightarrow ((\varphi \succeq g)\rightarrow B_a [\varphi \succeq g]\psi)$
	\end{enumerate}
\end{Proposition}
\begin{proof}
	We prove that each scheme is valid for an arbitrary D\L L-model $\mathfrak{M} = (S, r_{a_{|_{a\in \mathcal{A}}}}, \pi)$ and an arbitrary state $s\in S$.
	
	\textbf{(1):} We show that  $V_s([\varphi\succeq g]p \rightarrow ((\varphi\succeq g)\rightarrow p)) =1$.  If $V_s^\mathfrak{M}(\varphi \succeq g)=0$, then we have
	\begin{align*}
		& V_s^{\mathfrak{M}}([\varphi \succeq g]p \rightarrow (\varphi \succeq g \rightarrow p))=&\\
		& \min \{1, 1-V_s^{\mathfrak{M}}([\varphi \succeq g]p)+V_s^{\mathfrak{M}}(\varphi \succeq g \rightarrow p)\}=&\\
		&\min \{1, 1-1+1\}=1
	\end{align*}
	So assume $V_s^{\mathfrak{M}}(\varphi \succeq g)=1$. In this case we have $V_s^{\mathfrak{M}}([\varphi \succeq g]p) = V_s^{\mathfrak{M}^{\varphi \succeq g}}(p) = V_s^\mathfrak{M} (p)$ and $V_s^{\mathfrak{M}}((\varphi \succeq g)\rightarrow p )= V_s^\mathfrak{M}(p)$ and so $V_s([\varphi\succeq g]p \rightarrow ((\varphi\succeq g)\rightarrow p)) =1$. By a similar argument it can be shown that $V_s(((\varphi\succeq g)\rightarrow p) \rightarrow [\varphi\succeq g]p   ) =1$.
	
	
	\textbf{(2):} 
	Let $V_s^{\mathfrak{M}^{\varphi \succeq g}}(\varphi \succeq g)=1$.  We have
	\begin{align*}
		& V_s^\mathfrak{M}([\varphi \succeq g]\neg \psi \rightarrow((\varphi \succeq g) \rightarrow \neg [\varphi \succeq g]\psi)) =\\
		& \min\{1, 1- V_s^\mathfrak{M}([\varphi \succeq g]\neg \psi) + V_s^\mathfrak{M}((\varphi \succeq g) \rightarrow \neg [\varphi \succeq g]\psi)\} = \\
		& \min \{1, 1- V_s^{\mathfrak{M}^{\varphi \succeq g}}(\neg\psi) + V_s^\mathfrak{M}(\neg [\varphi\succeq g ]\psi)\} = \\
		& \min \{1, 1-1+ V_s^{\mathfrak{M}^{\varphi \succeq g}}(\psi) + 1- V_s^\mathfrak{M}([\varphi\succeq g ]\psi)\} = \\
		& \min \{1, 1- 1 + V_s^{\mathfrak{M}^{\varphi \succeq g}}(\psi) + 1- V_s^{\mathfrak{M}^{\varphi \succeq g}}(\psi)\} = 1.
	\end{align*}
	If $V_s^{\mathfrak{M}^{\varphi \succeq g}}(\varphi \succeq g)=0$, the statement is obviously obtained. The inverse  part is shown by a similar discussion.
	
	\textbf{(3):}
	We assume that $V_s^{\mathfrak{M}}(\varphi \succeq g) =1$, the other case is trivial. We have:
	\begin{align}
		& V_s^\mathfrak{M}([\varphi \succeq g]\psi \,\&\, \chi \rightarrow ([\varphi \succeq g]\psi \,\&\, [\varphi \succeq g]\chi)) = \nonumber \\
		&\min \{1, 1-V_s^{\mathfrak{M}}([\varphi \succeq g]\psi \,\&\, \chi) + V_s^\mathfrak{M}([\varphi \succeq g]\psi \,\&\, [\varphi \succeq g]\chi)\} = \nonumber\\ 
		&\min \{1, 1-V_s^{\mathfrak{M}^{\varphi \succeq g}}(\psi \,\&\, \chi) + \max\{0, V_s^\mathfrak{M}([\varphi \succeq g]\psi) + V_s^\mathfrak{M}([\varphi\succeq g]\chi)-1\}\} = \nonumber\\
		& \min \{1, 1-\max\{0, V_s^{\mathfrak{M}^{\varphi \succeq g}}(\psi ) + V_s^{\mathfrak{M}^{\varphi \succeq g}}(\chi) - 1\} + \max\{0, V_s^{\mathfrak{M}^{\varphi \succeq g}}(\psi) + V_s^{\mathfrak{M}^{\varphi \succeq g}}(\chi)-1\}\} = 1. \nonumber
	\end{align}
	$V_s^\mathfrak{M}( ([\varphi \succeq g]\psi \,\&\, [\varphi \succeq g]\chi)\rightarrow [\varphi \succeq g]\psi \,\&\, \chi )=1$ has  similar computations.\\
	
	\textbf{(4):}
	It is easy to check that $V_s(\psi \rightarrow \chi) = V_s(\neg (\psi \,\&\, \neg \chi))$. Thus the statement holds.
	
	\textbf{(5):} Let $V_s^{\mathfrak{M}^{\varphi \succeq g}}(\varphi \succeq g)=1$. We have:
	\begin{align}
		& V_s^\mathfrak{M}([\varphi \succeq g]B_a \psi \rightarrow ((\varphi \succeq g)\rightarrow B_a [\varphi \succeq g]\psi)) = \nonumber \\
		& \min \{1, 1-V_s^\mathfrak{M}([\varphi \succeq g]B_a \psi) + V_s^\mathfrak{M}((\varphi \succeq g)\rightarrow B_a[\varphi \succeq g]\psi)\} = \nonumber\\
		& \min \{1, 1-V_s^\mathfrak{M}([\varphi \succeq g]B_a \psi) + V_s^\mathfrak{M}(B_a[\varphi \succeq g]\psi)\} = \nonumber \\
		&\min \{1, 1-V_s^\mathfrak{M^{\varphi \succeq g}}(B_a \psi) + V_s^\mathfrak{M}(B_a[\varphi \succeq g]\psi)\} = \nonumber \\
		& \min \{1, 1-\inf_{s'\in S^{\varphi \succeq g}}\max\{1-r_a^{\varphi \succeq g}(s,s') , V_{s'}^\mathfrak{M^{\varphi \succeq g}}(\psi)\}
		+ \inf_{s''\in S^{\mathfrak{M}}}\max\{1-r_a^\mathfrak{M}(s, s''), V_{s''}^\mathfrak{M}([\varphi \succeq g]\psi)\}\} = \nonumber \\
		& \min \{1, 1-\inf_{s'\in S^{\varphi \succeq g}}\max\{1-r_a^{\varphi \succeq g}(s,s') , V_{s'}^\mathfrak{M^{\varphi \succeq g}}(\psi)\}
		+ \inf_{s''\in S^{\mathfrak{M}}}\max\{1-r_a^\mathfrak{M}(s, s''), V_{s''}^{\mathfrak{M}^{\varphi\succeq g}}(\psi)\}\} =1 \nonumber
	\end{align}
	Note that the last statement is followed from the fact that when $V_{s''}^{\mathfrak{M}^{\varphi \succeq g}}(\psi)=1$, for the states $s''$ in which $V_{s''}(\varphi\succeq g)=0$, then we have
	$$  \inf_{s'\in S^{\varphi \succeq g}}\max\{1-r_a^{\varphi \succeq g}(s,s') , V_{s'}^\mathfrak{M^{\varphi \succeq g}}(\psi)\} = \inf_{s''\in S^{\mathfrak{M}}}\max\{1-r_a^\mathfrak{M}(s, s''), V_{s''}^{\mathfrak{M}^{\varphi\succeq g}}(\psi)\}.$$
	$V_s^\mathfrak{M}( ((\varphi \succeq g)\rightarrow B_a [\varphi \succeq g]\psi)\rightarrow [\varphi \succeq g]B_a \psi ) =1$ has a similar argument.
\end{proof}

\begin{Definition}
	Let $\varphi, \psi$ and $\chi$ be \textbf{DD\L L}-formulae  and $a\in\mathcal{A}$. The following Dynamic \L ukasiewicz axiomatic system \textbf{D\L} is an extension of doxastic \L ukasiewicz logic \textbf{B\L}$^+$.
	\begin{itemize}
		\item[($\text{\L}_D0$)] all tautologies of Doxastic \L ukasiewicz logic \textbf{B\L}$^+$.
		\item[($\text{\L}_D1$)] $[\varphi \succeq g]p \leftrightarrow ((\varphi \succeq g) \rightarrow p)$
		\item[($\text{\L}_D2$)] $[\varphi \succeq g]\neg \psi \leftrightarrow ((\varphi \succeq g) \rightarrow \neg [\varphi \succeq g]\psi)$
		\item[($\text{\L}_D3$)] $[\varphi \succeq g](\psi \,\&\, \chi) \leftrightarrow ([\varphi \succeq g]\psi \,\&\, [\varphi \succeq g]\chi)$
		\item[($\text{\L}_D4$)] $[\varphi \succeq g] (\psi \rightarrow \chi) \leftrightarrow [\varphi \succeq g]\neg (\psi \,\&\, \neg\chi)$
		\item[($\text{\L}_D5$)]  $[\varphi \succeq g]B_a \psi \leftrightarrow ((\varphi \succeq g)\rightarrow B_a [\varphi \succeq g]\psi)$
	\end{itemize}
\end{Definition}

\begin{Theorem} \textbf{(Soundness)} \label{soundness}
	
	The axiomatic system \textbf{D\L} is sound,  i.e. if a \textbf{DD\L L}-formula $\varphi$ is provable in \textbf{D\L}, then it is valid in all D\L L-models.
\end{Theorem}
\begin{proof}
	It is a direct conclusion of Proposition \ref{valid_schema}.
\end{proof}

In order to prove the completeness theorem, we follow Ditmarsch's approach for public announcement in \cite{delDitmarsch}.

\begin{Definition}
	Let $\varphi, \psi$ and $\chi$ be \textbf{DD\L L}-formulae.
	A \textit{translation} between languages \textbf{DD\L L} and \textbf{D\L L} is a function $t:\textbf{DD\L L}\rightarrow \textbf{D\L L}$, which is defined recursively as follows:
	
	\begin{align*}
		& \forall p \in \mathcal{P} \quad t(p) = p,& & t(\neg \varphi) = \neg t(\varphi), \\
		& t(\varphi \,\&\, \psi) = t(\varphi) \,\&\, t(\psi), & & t(\varphi \rightarrow \psi) = t(\neg (\varphi \,\&\, \neg \psi)),\\
		&  t(B_a \varphi) = B_a t(\varphi), & & t([\varphi\succeq g]p) = t(\varphi \succeq g \rightarrow p), \\
		& t([\varphi\succeq g]\neg \psi) = t(\varphi \succeq g \rightarrow \neg [\varphi \succeq g]\psi), & & t([\varphi \succeq g](\psi \,\&\,\chi)) = t([\varphi \succeq \psi] \,\&\, [\varphi \succeq g]\chi), \\
		& t ([\varphi \succeq g] (\psi \rightarrow \chi)) = t ([\varphi \succeq g]\neg ( \psi \,\&\, \neg \chi))&
		& t([\varphi \succeq g] B_a \psi) = t(\varphi\succeq g \rightarrow B_a [\varphi \succeq g]\psi).
	\end{align*}
	The \textit{complexity} of a \textbf{DD\L L}-formula  is defined using the function $c:\textbf{DD\L L} \rightarrow \mathbb{N}$ as follows:
	\begin{align*}
		& c(p)=1, && c(\neg \varphi) = 1 + c(\varphi), && c(\varphi\succeq g)= 1+c(\varphi),\\
		&  c(\varphi \,\&\, \psi) = 1 + \max\{c(\varphi), c(\psi)\},&& c(\varphi\rightarrow \psi)= 3+\max\{c(\varphi), c(\psi)\}, && c(B_a \varphi) = 1 + c(\varphi)\\
		&  c([\varphi \succeq g]\psi) = (5+c(\varphi))c(\psi)
	\end{align*}
\end{Definition}

\begin{Lemma} \label{complexity}
	For all $\varphi, \psi$ and $\chi$ in \textbf{DD\L L} we have:
	\begin{enumerate}
		\item $c(\psi)\geq c(\varphi)$ if $\varphi$ is sub-formula of $\psi$.
		\item $c([\varphi\succeq g]p) > c((\varphi\succeq g )\rightarrow p)$
		\item $c([\varphi\succeq g]\neg \psi) > c((\varphi\succeq g) \rightarrow \neg [\varphi \succeq g]\psi)$
		\item $c([\varphi \succeq g](\psi \,\&\, \chi)) > c([\varphi \succeq g]\psi \,\&\, [\varphi \succeq g]\chi)$
		\item $c([\varphi \succeq g](\psi \rightarrow \chi)) > c([\varphi \succeq g]\neg (\psi \,\&\, \neg \chi))$
		\item $c([\varphi \succeq g]B_a \psi) > c(\varphi \succeq g \rightarrow B_a [\varphi \succeq g]\psi)$
	\end{enumerate}
\end{Lemma}

\begin{proof}
	\begin{itemize}
		\item[\textbf{(1):}] It is obvious by definition \ref{complexity}.
		\item [\textbf{(2):}] By definition we have $c([\varphi\succeq g]p) = (5+c(\varphi))c(p)$ which is equal to $5+c(\varphi)$. Also, we have $c(\varphi\succeq g \rightarrow p) = 3+\max\{c(\varphi\succeq g), c(p) \} = 4+c(\varphi)$. So it's obvious that $c([\varphi\succeq g]p) > c((\varphi \succeq p)\rightarrow p).$
		\item[\textbf{(3):}] We have:
		\begin{align}
			& c([\varphi\succeq g]\neg \psi)&& = (5+c(\varphi))(1+c(\psi)) \nonumber \\
			& && = 5 + c(\varphi) + 5c(\psi) + c(\varphi)c(\psi), \label{eq002}
		\end{align}
		\begin{align}
			& c((\varphi \succeq g)\rightarrow \neg [\varphi \succeq g]\psi)&& = 3+\max\{c(\varphi\succeq g), c(\neg[\varphi \succeq g]\psi)\} \nonumber\\
			&&&=4+ c([\varphi\succeq g]\psi)\nonumber\\
			&&&=4+(5+c(\varphi))c(\psi)\nonumber\\
			&&&=4+5c(\psi) + c(\varphi)c(\psi) \label{eq003}
		\end{align}
		It is easy to see that \ref{eq002} is greater than \ref{eq003}.
		\item[\textbf{(4):}] Without loss of generality let $c(\psi) > c(\chi)$. We have:
		\begin{align}
			&c([\varphi \succeq g](\psi\,\&\, \chi))&&=(5+c(\varphi))(1+\max\{c(\psi), c(\chi)\}) \nonumber \\
			& &&= 5 +c(\varphi) +  5c(\psi) + c(\varphi)c(\psi) \label{eq004}
		\end{align}
		\begin{align}
			&c([\varphi\succeq g]\psi \,\&\, [\varphi\succeq g]\chi) && =1+\max\{c([\varphi\succeq g]\psi), c([\varphi\succeq g]\chi)\}\nonumber\\
			&&&= 1+\max\{(5+c(\varphi))c(\psi), (5+c(\varphi)c(\chi))\}\nonumber\\
			&&&= 1+ 5c(\psi) + c(\varphi)c(\psi) \label{eq005}
		\end{align}
		From \ref{eq004}, \ref{eq005} it is easy to check that $c([\varphi\succeq g]\psi\,\&\, \chi) > c([\varphi\succeq g]\psi \,\&\, [\varphi\succeq g]\chi)$
		\item[\textbf{(5):}] Without loss of generality let $c(\psi)> c(\chi)$. We have:
		\begin{align}
			\nonumber& c([\varphi \succeq \chi] (\psi \rightarrow \chi)) && = (5+c(\varphi))(3+\max \{c(\psi), c(\chi)\})\\
			\nonumber& &&= (5+c(\varphi))(3+c(\psi)) \\
			& &&= 15 + 5c(\psi) + 3c(\varphi) + c(\varphi)c(\psi) \label{eq006}
		\end{align}
		\begin{align}
			\nonumber& c([\varphi \succeq g] \neg (\psi \,\&\, \neg \chi)) && = (5+c(\varphi))(2+\max \{c(\psi), 1+c(\chi)\})\\
			\nonumber & && = (5+c(\varphi))(2+c(\psi))\\
			& && = 10 + 5c(\psi) + 2c(\varphi) + c(\varphi)c(\psi) \label{eq007}
		\end{align}
		It can be seen that \ref{eq006} is greater than \ref{eq007}, and so the statement holds.
		\item[\textbf{(6):}] We have:
		\begin{align*}
			& c([\varphi \succeq g]B_a \psi) && = (5+c(\varphi))(1+c(\psi)) \\
			&&& = 5 + c(\varphi) + 5c(\psi) + c(\varphi)c(\psi) \\
			&&&> 4+ 5c(\psi) + c(\varphi)c(\psi) \\
			&&& >3+(1+ (5+c(\varphi))c(\psi))\\
			&&&\geq 3 + c(B_a [\varphi\succeq g]\psi)\\
			&&&= 3 + \max\{c(\varphi\succeq g), c(B_a[\varphi\succeq g]\psi)\}\\
			&&&=c((\varphi\succeq g)\rightarrow B_a[\varphi \succeq g]\psi)
		\end{align*}
	\end{itemize}
\end{proof}

\begin{Lemma} \label{translation_prove}
	For all formula $\varphi$ in \textbf{DD\L L} we have $\vdash_{\textbf{D\L}} \varphi \leftrightarrow t(\varphi)$.
\end{Lemma}
\begin{proof}
	Let $\varphi$ be a \textbf{DD\L L}-formula. The proof is performed by induction on $c(\varphi)$. The base case is when  $\varphi = p$, for some $p\in \mathcal{P}$. It is obvious that $\vdash_{\textbf{D\L}} p \leftrightarrow p$. So suppose that $c(\varphi)=n$, and the statement holds for all $\psi$ where $c(\psi)< n$.
	We have the following cases:
	\begin{itemize}
		\item \textbf{case $\neg\varphi$:} From Lemma \ref{complexity} item 1, we have $c(\neg \varphi)>c(\varphi)$. Thus by induction hypothesis we have $\vdash_{\textbf{D\L}} \varphi \leftrightarrow t(\varphi)$. Therefore by the following deduction we have $\vdash_{\textbf{D\L}} \neg \varphi \leftrightarrow \neg t(\varphi)$, and since $t(\neg \varphi)=\neg t( \varphi)$, hence $\vdash_{\textbf{D\L}}\neg \varphi \leftrightarrow t(\neg \varphi)$.
		$$
		\begin{array}{lll}
			(1) & (\varphi \rightarrow \psi) \,\&\, (\psi \rightarrow \varphi) & \text{induction hypothesis where } \psi = t(\varphi) \text{ and definition of } \leftrightarrow\\
			
			(2)& (\varphi \rightarrow \psi) \,\&\, (\psi \rightarrow \varphi) \rightarrow (\varphi \rightarrow \psi) & (A1) \\
			
			(3) & (\varphi \rightarrow \psi) \,\&\, (\psi \rightarrow \varphi) \rightarrow (\psi \rightarrow \varphi)& (A2), (A1) \\
			
			(4) & (\varphi \rightarrow \psi) & (1), (2), (R_{MP}) \\
			
			(5) & (\psi \rightarrow \varphi) & (1), (3), (R_{MP})\\
			
			(6) & (\varphi\rightarrow \psi)\rightarrow(\neg \psi \rightarrow \neg \varphi) & \text{theorem } (\L uka.)\\
			
			(7) & (\psi \rightarrow \varphi)\rightarrow (\neg \varphi \rightarrow \neg \psi) & \text{theorem } (\L uka.)\\
			
			(8) & (\neg \psi \rightarrow \neg \varphi) & (4), (6), (R_{MP}) \\
			
			(9) & (\neg \varphi \rightarrow \neg \psi) & (5), (7) , (R_{MP}) \\
			
			(10) & (\neg \psi \rightarrow \neg \varphi) \,\&\, (\neg \varphi \rightarrow \neg \psi) & (8), (9), \text{similar to \ref{and_Luka}}
		\end{array}
		$$
		
		\item \textbf{case $\varphi \,\&\, \psi$:} We have $c(\varphi \,\&\, \psi)\geq c(\varphi)$ and $c(\varphi \,\&\, \psi)\geq c(\psi)$ by Lemma \ref{complexity} item 1. Thus we have $\vdash_{\textbf{D\L}} \varphi \leftrightarrow t(\varphi)$ and $\vdash_{\textbf{D\L}} \psi \leftrightarrow t(\psi)$ from induction hypothesis. Now  $\vdash_{\textbf{D\L}} (\varphi \,\&\, \psi) \leftrightarrow (t(\varphi)\,\&\, t(\psi))$ by the following process:
		$$\begin{array}{lll}
			(1) & (\varphi\rightarrow t(\varphi))\,\&\,(t(\varphi)\rightarrow\varphi) & \text{ind. hypothesis}\\
			
			(2) & ((\psi\rightarrow t(\psi))\,\&\,(t(\psi)\rightarrow\psi)) & \text{ind. hypothesis}\\
			
			(3) & (\varphi\rightarrow t(\varphi))\,\&\,(t(\varphi)\rightarrow\varphi)\rightarrow (\varphi\rightarrow t(\varphi))& (A1) \\
			(4) & (\psi\rightarrow t(\psi))\,\&\,(t(\psi)\rightarrow\psi)\rightarrow (\psi\rightarrow t(\psi))& (A1) \\
			
			(5) & \varphi \rightarrow t(\varphi) & (1), (3), (R_{MP})\\
			
			(6) & \psi \rightarrow t(\psi) & (2), (4), (R_{MP}) \\
			
			(7) & ((\varphi \rightarrow t(\varphi)) \,\&\, (\psi \rightarrow t(\psi))) \rightarrow ((\varphi \,\&\, \psi)\rightarrow (t(\varphi)\,\&\, t(\psi)))& 
			(\L2)\\
			
			(8) & (\varphi \rightarrow t(\varphi)) \,\&\, (\psi \rightarrow t(\psi)) & (5), (6), \text{similar to \ref{and_Luka}}\\
			
			(9) &  (\varphi \,\&\, \psi)\rightarrow (t(\varphi)\,\&\, t(\psi)) & (7), (8), (R_{MP}) \\
			
			(10) & (t(\varphi)\,\&\, t(\psi)) \rightarrow (\varphi \,\&\, \psi) & \text{has a similar deduction as (9)}\\
			
			(11) & ((\varphi \,\&\, \psi)\rightarrow (t(\varphi)\,\&\, t(\psi))) \,\&\, ((t(\varphi)\,\&\, t(\psi)) \rightarrow (\varphi \,\&\, \psi))&  (9), (10), \text{similar to \ref{and_Luka}}
		\end{array}$$
		 Therefore, simply we obtain $\vdash_{\textbf{D\L}} \varphi\,\&\,\psi \leftrightarrow t(\varphi\,\&\, \psi)$ from translation definition.
		 \item \textbf{case $\varphi \rightarrow \psi$:} From Lemma \ref{complexity} we have $c(\varphi \rightarrow \psi) > c(\varphi)$ and $c(\varphi\rightarrow \psi)> c(\psi)$. So by induction hypothesis we have $\vdash_{\textbf{D\L}} \varphi \leftrightarrow t(\varphi)$ and $\vdash_{\textbf{D\L}} \psi \leftrightarrow t(\psi)$. By similar proof given in the previous case $\neg \varphi$, we have $\vdash_{\textbf{D\L}}\neg \psi \leftrightarrow \neg t(\psi)$. Thus we have:
		 $$\begin{array}{lll}
		 	(1) & (\varphi\rightarrow t(\varphi))\,\&\,(t(\varphi)\rightarrow\varphi) & \text{ind. hypothesis}\\
		 	
		 	(2) & ((\neg \psi\rightarrow \neg t(\psi))\,\&\,(\neg t(\psi)\rightarrow\neg \psi)) & \text{deduced from ind. hypothesis}\\
		 	
		 	(3) & (\varphi\rightarrow t(\varphi))\,\&\,(t(\varphi)\rightarrow\varphi)\rightarrow (\varphi\rightarrow t(\varphi))& (A1) \\
		 	(4) & (\neg \psi\rightarrow \neg t(\psi))\,\&\,(\neg t(\psi)\rightarrow\neg\psi)\rightarrow (\neg\psi\rightarrow \neg t(\psi))& (A1) \\
		 	
		 	(5) & \varphi \rightarrow t(\varphi) & (1), (3), (R_{MP})\\
		 	
		 	(6) & \neg \psi \rightarrow \neg t(\psi) & (2), (4), (R_{MP}) \\
		 	
		 	(7) & ((\varphi \rightarrow t(\varphi)) \,\&\, (\neg \psi \rightarrow \neg t(\psi))) \rightarrow ((\varphi \,\&\, \neg \psi)\rightarrow (t(\varphi)\,\&\, \neg t(\psi)))& 
		 	(\L2)\\
		 	
		 	(8) & (\varphi \rightarrow t(\varphi)) \,\&\, (\neg \psi \rightarrow \neg t(\psi)) & (5), (6), \text{similar to \ref{and_Luka}}\\
		 	
		 	(9) &  (\varphi \,\&\, \neg \psi)\rightarrow (t(\varphi)\,\&\, \neg t(\psi)) & (7), (8), (R_{MP}) \\
		 	
		 	(10) & (t(\varphi)\,\&\, \neg t(\psi)) \rightarrow (\varphi \,\&\, \neg \psi) & \text{has a similar deduction as (9)}\\
		 	
		 	(11) & ((\varphi \,\&\, \neg \psi)\rightarrow (t(\varphi)\,\&\, \neg t(\psi))) \rightarrow (\neg (t(\varphi)\,\&\, \neg t(\psi)) \rightarrow \neg (\varphi \,\&\, \neg \psi)) & \text{theorem} (\L uka.)\\
		 	
		 	(12) & ((t(\varphi)\,\&\, \neg t(\psi)) \rightarrow (\varphi \,\&\, \neg \psi)) \rightarrow (\neg (\varphi \,\&\, \neg \psi)\rightarrow \neg (t(\varphi)\,\&\, \neg t(\psi))) & \text{theorem} (\L uka.)\\
		 	
		 	(13) & \neg (t(\varphi)\,\&\, \neg t(\psi)) \rightarrow \neg (\varphi \,\&\, \neg \psi) & (9), (11), (R_{MP})  \\
		 	
		 	(14) & \neg (\varphi \,\&\, \neg \psi)\rightarrow \neg (t(\varphi)\,\&\, \neg t(\psi)) & (10), (12), (R_{MP}) \\
		 	
		 	(15) & (\neg (\varphi \,\&\, \neg \psi)\rightarrow \neg (t(\varphi)\,\&\, \neg t(\psi))) \,\&\, (\neg (t(\varphi)\,\&\, \neg t(\psi)) \rightarrow \neg (\varphi \,\&\, \neg \psi))&  (13), (14), \text{similar to \ref{and_Luka}}
		 \end{array}$$
	 Therefore we have $\vdash_{\textbf{D\L}} \neg (\varphi \,\&\, \neg \psi) \leftrightarrow \neg (t(\varphi) \,\&\, \neg t(\psi))$, and by definition we have 
	 $\vdash_{\textbf{D\L}} \neg (\varphi \,\&\, \neg \psi) \leftrightarrow t(\neg (\varphi\,\&\,\neg \psi))$. Meanwhile in \L ukasiewicz logic we have $\vdash_{\textbf{D\L}} (\varphi\rightarrow \psi)\leftrightarrow\neg (\varphi \,\&\, \neg \psi)$. So we have $\vdash_{\textbf{D\L}} (\varphi\rightarrow \psi)\leftrightarrow t(\neg (\varphi\,\&\,\neg \psi))$, and then by transition definition we obtain $\vdash_{\textbf{D\L}} (\varphi\rightarrow \psi)\leftrightarrow t(\psi \rightarrow \varphi)$ as desired.
		 
		\item \textbf{case $[\varphi \succeq g] p$:} We have $c([\varphi \succeq g]p)> c((\varphi \succeq g)\rightarrow p)$ from Lemma \ref{complexity} item 2. Thus from induction hypothesis we have $\vdash_{\textbf{D\L}} ((\varphi \succeq g)\rightarrow p)\leftrightarrow t((\varphi \succeq g)\rightarrow p)$, and  $\vdash_{\textbf{D\L}} ((\varphi\succeq g)\rightarrow p) \leftrightarrow t([\varphi \succeq g]p)$ is obtained by translation definition, then from (\L$_D$1) we have $\vdash_{\textbf{D\L}} ([\varphi \succeq g]p) \leftrightarrow t([\varphi \succeq g]p)$  straightforward.
		\item \textbf{case $[\varphi \succeq g]\neg \psi$:} From Lemma \ref{complexity} item 3 we have $c([\varphi \succeq g]\neg \psi) > c((\varphi \succeq g \rightarrow \neg [\varphi \succeq g]\psi))$. Thus by induction hypothesis we have $\vdash_{\textbf{D\L}} ((\varphi \succeq g)\rightarrow \neg [\varphi \succeq g]\psi)\leftrightarrow t((\varphi \succeq g)\rightarrow \neg [\varphi \succeq g]\psi)$. Hence  using translation definition we have $\vdash_{\textbf{D\L}} ((\varphi\succeq g)\rightarrow \neg [\varphi \succeq g]\psi) \leftrightarrow t([\varphi \succeq g]\neg \psi)$. Then by (\L$_D$2) we obtain $\vdash_{\textbf{D\L}} ([\varphi \succeq g]\neg \psi) \leftrightarrow t([\varphi \succeq g]\neg \psi)$. 
		\item \textbf{case $[\varphi \succeq g](\psi \,\&\,\chi)$:} We have $c([\varphi \succeq g](\psi \,\&\,\chi))\geq c([\varphi \succeq g]\psi \,\&\, [\varphi \succeq g]\chi)$ by Lemma \ref{complexity} item 4. Therefore by induction hypothesis we have $\vdash_{\textbf{D\L}} ([\varphi \succeq g]\psi \,\&\, [\varphi \succeq g]\chi)\leftrightarrow t([\varphi \succeq g]\psi \,\&\, [\varphi \succeq g]\chi)$, and similar to the previous cases using  (\L$_D$3) and translation definition, it is obtained that $\vdash_{\textbf{D\L}} [\varphi \succeq g]\psi \,\&\, \chi \leftrightarrow t([\varphi \succeq g]\psi \,\&\,\chi)$.
		\item \textbf{case $[\varphi \succeq g] (\psi \rightarrow \chi)$:} From Lemma \ref{complexity}, item 5, we have $c([\varphi \succeq g] (\psi \rightarrow \chi)) > c([\varphi \succeq g] \neg (\psi \,\&\, \neg \chi))$. So by induction hypothesis we have $\vdash_{\textbf{D\L}} [\varphi \succeq g] \neg (\psi \,\&\, \neg \chi) \leftrightarrow t([\varphi \succeq g]\neg (\psi \,\&\, \neg \chi))$. From (\L$_D$4) we obtain $\vdash_{\textbf{D\L}} [\varphi \succeq g] (\psi \rightarrow \chi) \leftrightarrow t( [\varphi \succeq g] \neg (\neg \psi \,\&\, \chi))$, so by translation definition of $t([\varphi\succeq g]\neg (\psi\,\&\, \neg \chi))$ we have $\vdash_{\textbf{D\L}} [\varphi \succeq g] \psi \rightarrow \chi \leftrightarrow t([\varphi \succeq g] (\psi \rightarrow \chi))$.
		\item \textbf{case $[\varphi \succeq g]B_a \psi$:} We have $c([\varphi\succeq g]B_a \psi)>c((\varphi\succeq g )\rightarrow B_a [\varphi\succeq g]\psi)$ by Lemma \ref{complexity} item 6. Thus by induction hypothesis we have $\vdash_{\textbf{D\L}}((\varphi\succeq g )\rightarrow B_a [\varphi\succeq g]\psi)\leftrightarrow t((\varphi\succeq g )\rightarrow B_a [\varphi\succeq g]\psi)$. Again similar to the precious cases  using (L$_D$5) and translation definition of $t([\varphi \succeq g]B_a \psi)$ we obtain $\vdash_{\textbf{D\L}} [\varphi \succeq g]B_a\psi \leftrightarrow t([\varphi\succeq g]B_a \psi)$.
	\end{itemize}
\end{proof}

\begin{Theorem}
	\textbf{(Completeness)}\label{thm_pa_compeleteness}
	
	For each formula $\varphi$ in \textbf{DD\L L}, $\vDash \varphi$ implies $\vdash_{\textbf{D\L}} \varphi$.
\end{Theorem}
\begin{proof}
	Suppose that $\vDash \varphi$. By translation definitions we have $\vDash t(\varphi)$. Thus by completeness of \textbf{B\L}$^+$ we have $\vdash_{\textbf{B\L}^+} t(\varphi)$. Since \textbf{B\L}$^+$ is a fragment of \textbf{D\L}, we have $\vdash_{\textbf{D\L}} t(\varphi)$. Also, by Lemma \ref{translation_prove} we have $\vdash_{\textbf{D\L}} \varphi \leftrightarrow t(\varphi)$.  Then, by $\vdash_{\textbf{D\L}} \varphi \leftrightarrow t(\varphi)$ and $\vdash_{\textbf{D\L}} t(\varphi)$ we can conclude that $\vdash_{\textbf{D\L}}\varphi$.
\end{proof}

\section{Conclusion}

In this paper by defining a new operator $\succeq$, we proposed an extension $\textbf{B\L}^+$ of doxastic \L ukasiewicz logic with pseudo-classical belief defined in \cite{Dll2022}.  We showed that $\textbf{B\L}^+$ is sound and complete with respect to the class of D\L L$^{}$-models. Also, we equipped  the language of $\textbf{B\L}^+$  with a public announcement operator $[.\succeq g]$ and introduced a dynamic extension $\textbf{D\L }$ of doxastic \L ukasiewicz logic. 
	Finally, we proved that $\textbf{D\L }$ is  sound and complete.


\end{document}